\PassOptionsToPackage{usenames,dvipsnames}{xcolor}

\documentclass[runningheads, envcountsame]{llncs}

\usepackage{xcolor}
\usepackage{amsmath,amsfonts,amssymb,amscd}
\usepackage{latexsym}
\usepackage{xspace}
\usepackage{nicefrac}

\usepackage{todonotes}
\newcommand{\il}[1]{\todo[inline, color= yellow]{#1}}
\newcommand{\ma}[1]{\todo[color= SpringGreen]{\tiny{#1}}}

\newcommand{\Ma}[1]{{\textcolor{magenta}{#1}}}

\usepackage{hyperref}
\usepackage{cleveref}
\usepackage{float}
\usepackage[inline]{enumitem}
\usepackage{algorithm, algpseudocode}

\newtheorem{observation}{Observation}
\newtheorem{newclaim}[theorem]{Claim}
\crefname{observation}{observation}{observations}
\crefname{newclaim}{claim}{claims}
\Crefname{observation}{Observation}{Observations}
\Crefname{newclaim}{Claim}{Claims}

\allowdisplaybreaks

\makeatletter
\newcommand\cdotfill{%
    \leavevmode\cleaders\hb@xt@.44em{\hss$\cdot$\hss}\hfill\kern\z@
}
\makeatother

\newcommand{\bigoh}{\mathcal{O}}

\newcommand{\shortversion}[1]{}

\newcommand{\hide}[1]{}

\newcommand{\FVSbIN}{{\sf DFVS-bIN}\xspace}
\newcommand{\wFVSbIN}{{\sf FVS-bIN}\xspace}
\newcommand{\DFVSbIN}{{\sf DFVS-bIN}\xspace}
\newcommand{\sfvst}{{\sf S-FVST}\xspace}
\newcommand{\FVS}{{\rm FVS}\xspace}
\newcommand{\SFVS}{{\rm SFVS}\xspace}

\usepackage{hyperref}
\hypersetup{
    colorlinks=true,
    linkcolor=blue,
    filecolor=magenta,      
    urlcolor=cyan,
    }

\algrenewcommand\alglinenumber[1]{\scriptsize #1:}
\newenvironment{algo}
 {\par\addvspace{\topsep}
  \centering
  \begin{minipage}{\linewidth}
  \hrule\kern2pt}
 {\par\kern2pt\hrule
  \end{minipage}
  \par\addvspace{\topsep}}

\usepackage[misc]{ifsym}
\begin{document}
\title{Quick-Sort Style Approximation Algorithms for Generalizations of Feedback Vertex Set in Tournaments}
\titlerunning{Quick-Sort Style Approximation Algorithms for Generalizations of FVST}

\author{Sushmita Gupta\inst{1} \and Sounak Modak\inst{1}\textsuperscript{(\Letter)}
 \and Saket Saurabh\inst{1,2} \and Sanjay Seetharaman\inst{1}}

 \institute{The Institute of Mathematical Sciences, Chennai, India 
 \email{\{sushmitagupta,sounakm,saket,sanjays\}@imsc.res.in}
 \and
University of Bergen, Bergen, Norway
}

\authorrunning{S. Gupta et al.}
\maketitle

\begin{abstract}
    A feedback vertex set (FVS) in a digraph is a subset of vertices whose removal makes the digraph acyclic. In other words, it hits all cycles in the digraph. Lokshtanov et al. [TALG '21] gave a factor 2 randomized approximation algorithm for finding a minimum weight FVS in tournaments. We generalize the result by presenting a factor $2\alpha$ randomized approximation algorithm for finding a minimum weight FVS in digraphs of independence number $\alpha$; a generalization of tournaments which are digraphs with independence number $1$. 
    Using the same framework, we present a factor $2$ randomized approximation algorithm for finding a minimum weight Subset FVS in tournaments: given a vertex subset $S$ in addition to the graph, find a subset of vertices that hits all cycles containing at least one vertex in $S$. Note that FVS in tournaments is a special case of Subset FVS in tournaments in which $S = V(T)$.
\end{abstract}
\nocite{lokshtanov20212}

\section{Introduction}
Quicksort is a randomized divide-and-conquer algorithm for sorting a list of numbers. 
In this we randomly pick a pivot, partition the rest of the list into two parts, and recursively solve the two parts. 
The choice of pivot determines the size of the subproblems and consequently the overall running time. This forms a central idea in the polynomial time factor 2 randomized approximation algorithm for finding a minimum weight feedback vertex set in tournaments, given by  Lokshtanov et al. \cite{lokshtanov20212}. 
In this paper one of our goals is to find further problems for  which this approach can be applied in designing approximation algorithms. Without further ado we formally define the problem studied by Lokshtanov et al. \cite{lokshtanov20212}, which we generalize in this article. 

A tournament $T=(V,E)$ is a digraph in which there is exactly one arc between each pair of vertices (an orientation of a clique). A feedback vertex set is a subset of vertices whose removal makes the digraph acyclic. In the \textsc{Feedback Vertex Set in Tournaments} (FVST) problem, we are given a tournament $T$ and a weight function $w: V(G) \rightarrow \mathbb{N}_{\ge 0}$. 
The goal is to find a minimum weight subset of vertices whose removal makes the digraph acyclic.
It is a folklore that a tournament is acyclic if and only if it has no triangles (directed cycles of length $3$). This together with the local ratio technique \cite{bar2004local} gives a simple factor $3$  approximation algorithm for FVST. 
Cai et al. \cite{cai2001approximation} gave the first improvement over this algorithm and designed a $2.5$-approximation algorithm based on the total dual integral system combined with the local ratio technique.
Years later, Mnich et al. \cite{mnich20167} gave a $\nicefrac73$ approximation algorithm using the iterative rounding technique. Finally in 2020, Lokshtanov et al. \cite{lokshtanov20212} gave a ``quicksort style" randomized $2$-approximation algorithm. It is optimal (i.e. there is no polynomial time approximation algorithm with a better factor) assuming the Unique Games Conjecture \cite{khot2008vertex}. 

In this paper we apply this methodology to two generalizations of FVST.
\begin{enumerate}[wide=0pt]
    \item \textbf{Beyond Tournaments.} The {\it independence number} of a digraph is the size of a largest independent set in it. In \cite{seymour15}, Fradkin and Seymour introduced the class of digraphs of bounded independence number as a generalization of tournaments (which are digraphs with independence number 1). Problems studied in such digraphs include $k$-\textsc{Edge Disjoint Paths} \cite{seymour15}, \textsc{Edge Odd Cycle Transversal}, and \textsc{Feedback Arc Set} \cite{lochet2020fault}. In the \textsc{Directed FVS in graphs with Bounded Independence Number} (\FVSbIN) problem, we are given a digraph $G$ with independence number $\alpha$, called {\it $\alpha$-bounded digraph}, and a weight function $w: V(G) \rightarrow \mathbb{N}_{\ge 0}$. We are interested in finding a minimum weight feedback vertex set in $G$.
    
    \item \textbf{Subset Version of FVST.} A feedback vertex set is equivalently defined as a subset of vertices that hits all cycles in a digraph. In the \textsc{Subset Feedback Vertex Set in Tournaments} (\sfvst) problem, we are given as subset of vertices $S$ (often called as \textit{terminal set}) as input, in addition to a tournament $T$ and a weight function $w: V(T) \rightarrow \mathbb{N}_{\ge 0}$. We are interested in finding a minimum weight subset of vertices that hits all cycles that contain a vertex in $S$. Note that FVST is a special case of \sfvst in which $S = V(T)$.
\end{enumerate}
We design non-trivial randomized approximation algorithms for both \sfvst and \FVSbIN, with the techniques used in \cite{lokshtanov20212} as a starting point. Note that neither problem is a special case of the other. In \FVSbIN, we are generalizing the tournament graph by allowing the independence number to be $\alpha$ and in \sfvst, we are generalizing the obstruction set to be cycles that contain a vertex from $S$. Thus, the algorithm for one does not seem to apply for the other in a straightforward manner.

\subsection{Our Results}
Our first result is a randomized $2\alpha$-approximation algorithm for \FVSbIN. Observe that $G$ has a cycle if and only if $G$ has a cycle of length at most $2\alpha+1$. Due to this fact, we can apply the local ratio technique and obtain a $(2\alpha+1)$-approximate feedback vertex set in polynomial time. Improving upon this, we present the following.

\begin{theorem}
    \label{thm:fvsbin}
    \FVSbIN admits a randomized $2\alpha$-approximation algorithm that runs in time $n^{\bigoh(\alpha^2)}$.
\end{theorem}

Observe that for $\alpha =1$, we get a $2$-approximation, which is the case of tournaments. 
Note that the running time is a polynomial for fixed $\alpha$. A natural question that follows from our result is whether there exists an approximation algorithm that runs in time $f(\alpha)n^{\bigoh(1)}$ for some computable function $f$ (i.e., for a fixed $\alpha$, the running time is a polynomial in $n$ whose degree does not depend on $\alpha$).

Next, we present a randomized $2$-approximation algorithm for the \sfvst problem. We observe that for any vertex $s \in S$, $T$ has a cycle containing $s$ if and only if $T$ has a triangle containing $s$. Due to this fact, we can apply the local ratio technique and obtain a $3$-approximate subset feedback vertex in polynomial time. Improving upon this, we present the following. 

\begin{theorem}
    \label{thm:sfvst}
    \sfvst admits a randomized $2$-approximation algorithm that runs in time $n^{\bigoh(1)}$.
\end{theorem}

The proof of this theorem uses ideas similar to that of \Cref{thm:fvsbin}, with a twist. 
In both the results, we carefully exploit the fact that the size of the ``obstructions" (to a vertex not being part of a cycle) is bounded. 

\subsection{Our Methodology}
\label{subsec:our_methodology}
In this section we describe our main ideas in proving Theorems~\ref{thm:fvsbin} and \ref{thm:sfvst}. For simplicity of presentation we present our ideas for the unweighted case. 
To describe our methods, we first give a very brief outline of the algorithm of  Lokshtanov et al. \cite{lokshtanov20212}. This algorithm starts with the assumption that the tournament $T$ has a feedback vertex set $F$ of size at most $n/2$ (else returning all the vertices of the tournament is a $2$-approximation). Observe that if $F$ is a minimum feedback vertex set of $T$, then $T-F$ is acyclic and has a unique topological ordering. Next it selects a vertex $v$ uniformly at random. The probability that $v$ belongs to the  middle $n/6$ vertices of the topological ordering of $T-F$ is at least $1/6$. This $v$ will act as a pivot. The algorithm first deletes all the vertices that participate in directed triangles with $v$ (of course not $v$!). Since, it is known that $v$ does not belong to $F$, this step can be carried out at the cost of factor $2$ in approximation. After this step it is guaranteed that there is no directed triangle containing $v$ and hence the problem decomposes into two disjoint subproblems of size $\beta n$ ($\beta <1$ a constant): one on the tournament  induced by the in-neighbors of $v$ (i.e., $T[N^{-}(v)]$) and the other on the tournament induced by the out-neighbors of $v$ (i.e., $T[N^{+}(v)]$). The algorithm recursively solves these  problems and combines their solutions to obtain a $2$-approximation for $T$.


To generalize the techniques from tournaments to $\alpha$-bounded digraphs, our first challenge is to come up with a notion equivalent to the {\em unique topological ordering of an acyclic tournament}, that was exploited in the algorithm of Lokshtanov et al. \cite{lokshtanov20212}. 

Towards this we first observe that
every $\alpha$-bounded digraph $G$ on $n$ vertices has a vertex with {\em both in-degree and out-degree (at least) $\nicefrac{(n-2\alpha)}{4\alpha}$} (\Cref{clm:hiinhioutver}). Thereafter we define an {\sf HL-degree ordering} (HL stands for high to low) of an $\alpha$-bounded digraph $G$ as $\sigma(G) = \langle v_1, \dots, v_n \rangle$ over the vertices of $G$ by the recursive application of \Cref{clm:hiinhioutver}. Thus, for $v_1$ both $d^{+}_G(v_1)$ and $d^{-}_G(v_1)$ are at least $\nicefrac{(n-2 \alpha)}{4 \alpha}$; for $v_2$ both $d^{+}_{G - \{v_1\}}(v_2)$ and $d^{-}_{G - \{v_1\}}(v_2)$ at least $\nicefrac{(n-1-2 \alpha)}{4 \alpha}$. Therefore for any $i\in \{2,\dots, n\}$, both $d^{+}_{G -  \{v_1,\dots ,v_{i-1}\}}(v_i)$ and $d^{-}_{G -  \{v_1,\dots ,v_{i-1}\}}(v_i)$ are at least $\nicefrac{(n- (i-1) - 2 \alpha)}{4 \alpha}$. 

For our algorithm, given an $\alpha$-bounded digraph $G$, we will work with a fixed feedback vertex set $F$ and a fixed {\sf HL-degree ordering} $\sigma(G-F)$.  Analogous to the tournament's analysis, we can assume that $|F|\leq \frac{n}{2\alpha}$ (else the whole vertex set of $G$ is a $2\alpha$-approximation). Thereafter, we select a vertex $v$ uniformly at random. We say that $v$ is {\em good} if $v$ belongs to the first $(1/3)^{rd}$ part of $\sigma(G-F)$. The probability of $v$ being good is $\frac{\nicefrac{(n-\frac{n}{2\alpha})}{3}}{n}= (\frac13-\frac{1}{6\alpha})$. This vertex $v$ acts as our pivot. Now, we either delete all the vertices that participate in directed cycles of length at most $2\alpha+1$ with $v$ or delete all the vertices of a directed cycle that is part of the large cycle (i.e., length more than $2\alpha +1$ ) containing $v$, of length at most $2\alpha$
(follows from \Cref{obs:shortinducedcycles}). 
Suppose that $v$ is good. Then this step can be carried out at the cost of factor $2\alpha$ in approximation. 
After this step we are guaranteed that there is no directed cycle  containing $v$, and hence the problem decomposes into two disjoint subproblems  of size $\beta n$, where $\beta <1$ is a constant depending on $\alpha$ alone: one on the digraph induced by the vertices reachable from $v$ (denoted by $G[R_{G - P}(v)]$) and the other on the digraph induced by the vertices not reachable from $v$ (denoted by $G[\overline{R}_{G - P}(v)]$)
, where $P$ is a subset of vertices in $G$ such that the graph $G - P$ does not contain a cycle that contains vertex $v$. We recursively solve these problems, and combine their solutions to obtain a $2\alpha$-approximation for $G$.

For the \sfvst problem, given a tournament $T$ and terminal set $S$, we again work with a fixed subset feedback vertex set $F$ and an {\sf HL-degree ordering} $\sigma(T[S\setminus F])$. But here observe that we only focus on the vertices of $S$, that is, we want to see how $F$ interacts with $S$. First we have observed that hitting the triangles containing vertices in $S$ is equivalent to hitting the cycles that contain the vertices of $S$ (\Cref{prop:shortcycleinsfvst}).
As before, our algorithm starts with the assumption that $|F| \le |S|/2$ (else returning all the vertices of $S$ gives a $2$-approximation). Next it selects a vertex $v$ uniformly at random from $S$. The probability that $v$ belongs to the first $(1/3)^{rd}$ of $\sigma(T[S-F])$ is at least $1/6$. This vertex $v$ acts as a pivot. We first delete all the vertices that participate in directed triangles with $v$. Let the set of vertices which participate in directed triangles with $v$ be denoted by $P$. Since we know that $v$ does not belong to $F$, this step can be carried out at the cost of factor $2$ in approximation. After this step we are guaranteed that there is no directed triangle containing $v$, and hence there is no directed cycle containing $v$. Thereafter the problem decomposes into two disjoint subproblems with strictly smaller set of terminals: one on the tournament induced by $R_{T-P}(v)$ with terminal set $S\cap R_{T-P}(v)$ and the other on the tournament induced by $\overline{R}_{T-P}(v)$ with terminal set $S \cap \overline{R}_{T-P}(v)$. We recursively solve these  problems, and combine their solutions to obtain a $2$-approximation for $T$.

\section{Preliminaries}
In this paper, we deal with simple directed graphs (\textit{digraphs}, in short) containing no parallel edges/arcs.
We work in the setting of vertex weighted digraphs:  by $(G,w)$ we denote a vertex weighted digraph $G$ with weight function $w: V(G)\rightarrow \mathbb{N}_{\ge 0}$. The weight of a subset of vertices is the sum of weights of the vertices in the subset. Note that the setting of unweighted graphs is a special case of weighted graphs.
For any induced subgraph $H$ of a vertex weighted graph $(G,w)$, we will assume that $w$ defines a weight function when restricted to $V(H)$.
If there is an arc $(u,v) \in E(G)$, then $u$ is an \textit{in-neighbor} of $v$, and $v$ is an \textit{out-neighbor} of $u$.  
The \textit{in-neighborhood} of a vertex $x$, denoted by $N^{-}(x)=\{v\mid (v,x)\in E(G)\}$, is the set of in-neighbors of $x$. 
The \textit{in-degree} of a vertex $x$, denoted by $d^{-}_{G}(x)=|N^{-}(x)|$, is the number of in-neighbors of $x$.
The \textit{out-neighborhood} and \textit{out-degree} of a vertex $x$, denoted by $N^{+}(x)$ and $d^{+}_{G}(x)$ resp., are defined analogously. 


We say a vertex $u$ is \textit{reachable} from a vertex $v$, if there is a directed path which contains both the vertices $u$ and $v$ and vertex $v$ appears before vertex $u$ in the sequence of the vertices which defines the directed path. 
By $R_{G}(x)$, we denote the set of all vertices other than $x$ reachable from $x$ in $G$. 
By $\overline{R_{G}}(x)$, we denote the set of all vertices not reachable from $x$ in $G$.
Observe that $(R_{G}(x), \overline{R_{G}}(x), \{x\})$ form a partition of $V(G)$.

A \textit{feedback vertex set} (FVS) in $G$ is a subset of vertices $S\subseteq V(G)$ such that $G-S$ is acyclic. 
Given a family $X=\{S_1, \ldots , S_l\}$ where $S_i\subseteq V(G)$ for each $i\in [l]$, we call $S_a\in X$ \textit{lightest} if for all $i\in [l]$ we have $w(S_a)\leq w(S_i)$. Similarly, we call $S_b\in X$ \textit{heaviest} if for all $i\in [l]$ we have $w(S_i)\leq w(S_b)$. An FVS $F_{opt}$ in $G$ is an \textit{optimal} (also \textit{minimum}) solution of the instance $(G,w)$ if for every other FVS $S$ in $G$ we have $w(S)\geq w(F_{opt})$ i.e. $F_{opt}$ is the lightest among all FVSs in $G$. 
An FVS $F$ in $G$ is called a \textit{$2\alpha$-approximate solution} of the instance $(G,w)$ if $w(F)\leq 2\alpha\cdot w(F_{opt})$.

An algorithm is a \textit{randomized factor $f$ approximation algorithm} for a problem $\mathcal{P}$ if, for each instance $\mathcal{I}$ of $\mathcal{P}$, with probability at least $\nicefrac12$, it returns a solution for $\mathcal{I}$ of weight at most $f\times OPT_{\mathcal{I}}$ where $OPT_{\mathcal{I}}$ denotes the weight of an optimal solution for $\mathcal{I}$.

We will use the following structural results on digraphs with bounded independence number throughout our paper\footnote{Missing proofs (marked with $\dagger$) are in the full version of the paper}.



    




\begin{lemma}[\cite{seymour15}]
\label{lem: highinout}
Let $G$ be a simple digraph with $n$ vertices and independence number $\alpha\geq 1$. Then there exists a vertex in $V (G)$ with out-degree at least $\frac{n-\alpha}{2\alpha}$. Similarly, there exists a vertex
in $V (G)$ with in-degree at least $\frac{n-\alpha}{2\alpha}$.
\end{lemma}

\begin{lemma}[$\dagger$]
    \label{clm:hiinhioutver}
    Let $G$ be a simple digraph with $n$ vertices and independence number $\alpha\geq 1$. Then there exists a vertex in $V(G)$ which has both in-degree and out-degree\hide{$d^{+}(v)$ and $d^{-}(v)$ are} at least $\nicefrac{(n - 2 \alpha)}{4 \alpha}$.
\end{lemma}


\begin{lemma}
\label{obs:shortinducedcycles}

Let $G$ be a simple digraph with $n$ vertices and independence number $\alpha\geq 1$. For a vertex $x \in V(G)$, let $C \subseteq V(G)$ denote the shortest cycle containing $x$. Then either $|C| \le 2 \alpha + 1$ or there exists an induced cycle $C' \subset C$ of length at most $2 \alpha$ that does not contain $x$ (i.e., $C' \subseteq C \setminus \{x\}$).
\end{lemma}

\begin{proof}
Consider the case $|C| > 2 \alpha + 1$. Assume for the sake of contradiction that $C' = \langle v_1, \dots, v_\ell \rangle$ is the shortest induced cycle in $C$ that does not contain $x$ with $\ell \ge 2 \alpha + 1$. Additionally, assume without loss of generality that $C'$ is enumerated in such a way that it appears in that order in $C$ and $x$ appears before $v_1$ and after $v_\ell$ in $C$.

There is no arc $x v_j \in E(G)$ such that $v_j \in C'\setminus \{v_1\}$ since such an arc would imply the existence of the cycle $\langle v_j, \dots, v_\ell, \dots, x \rangle$ which contradicts the fact that $C$ is the shortest cycle in $G$ containing $x$. Since $C'$ is an induced cycle, the set $I = \{v_2, v_4, \dots, v_{2 \alpha}\}$ forms an independent set of size $\alpha$. By the previous argument $I \cup \{x\}$ is an independent set of size $\alpha+1$ in $G$, a contradiction. 
\qed
\end{proof}





\section{DFVS in Graphs of Bounded Independence Number}







Throughout this section we assume that $G$ is a digraph on $n$ vertices with independence number $\alpha$. The following two observations about FVSs in $(G,w)$ which we will use throughout our results in this section, follow from the hereditary property of acyclicity of digraphs. 

\begin{observation}\label{obs1}
    Let $F$ be an FVS in $(G,w)$; and let $S\subseteq V(G)$. Then, $F\setminus S$ is an FVS in $(G-S,w)$.
\end{observation}

\begin{observation}\label{obs2}
Suppose that $F$ is an optimal FVS in $(G,w)$ and that $S$ is a subset of $F$. Then, $F \setminus S$ is an optimal FVS in $(G - S,w)$, of weight $w(F)-w(S)$.
\end{observation}

The following lemma shows the interaction of an FVS $F$ in $G$ with the subgraphs $G[R_{G}(x)]$ and $G[\overline{R}_G(x)]$.

\begin{lemma}[$\dagger$]\label{lem:divide-conquer-FVS}\hide{Suppose that $(G,w)$ is an instance of \wFVSbIN.} Suppose that $x \in V(G)$ is a vertex that is not part of any cycle in $G$, then the following holds: $F$ is an FVS in $G$ if and only if $F\cap R_{G}(x)$ is an FVS in $G[R_{G}(x)]$ and $F\cap \overline{R}_G(x)$ is an FVS in $G[\overline{R}_G(x)]$.
\end{lemma}



\subsection{Technical Overview}

In this section we will briefly describe the work flow of our recursive algorithm FindFVS (\Cref{alg:findfvs}) which is formally presented in the following section and its correctness  analyzed in \Cref{ss:Analysis}. 
For the base case $n\le 30\alpha$, we compute an optimal FVS in $n^{\bigoh(\alpha)}$ time by checking all subsets of vertices. Let $F_{opt}$ be an optimal solution for $(G,w)$ which is an instance for the \DFVSbIN problem. We consider the following two cases:

\begin{itemize}[wide=0pt]
    \item{\textbf{Case 1 ($|F_{opt}|\ge \nicefrac{2n}{3\alpha}$):}} Let $L\subseteq V(G)$ be a set of $\nicefrac{n}{6\alpha}$ lightest vertices in $G$ and $r$ be the heaviest vertex in $L$. We define the new weight function $w^\prime: V(G)\setminus L \rightarrow \mathbb{N}_{\geq 0}$ which assigns the weight $w(v) - w(r)$ to each vertex $v$ in $G-L$. Our algorithm recursively finds an FVS $F$ in $(G- L,w^\prime)$. Combining $F$ with $L$ (i.e., $F \cup L$), we have a $2\alpha$-approximate FVS in $(G,w)$ with probability at least $\nicefrac{1}{2}$ (by \Cref{clm:largeopt}). 

    \item{\textbf{Case 2 ($|F_{opt}|< \nicefrac{2n}{3\alpha}$):}} 
    In this case, we randomly pick/sample a vertex $x\in V(G)$ and find a subset of vertices $P$ such that in $G - P$ there is no cycle that contains $x$. Then, since any cycle in $G - P$ is contained completely inside exactly one of $G[R_{G-P}(x)]$ and $G[\overline{R}_{G - P}(x)]$, we obtain an $FVS$ in $G$ by combining $P$ with FVSs in those two subgraphs. 
    
    With probability at least $\nicefrac{(n-\frac{2n}{3\alpha})}{n}= (1-\frac{2}{3\alpha})$, $x$ is not part of $F_{opt}$, i.e., $x \in V(G) \setminus F_{opt}$. Moreover, with probability at least $\nicefrac{(1-\frac{2}{3\alpha})}{3}=(\frac13-\frac{2}{9\alpha})$, $x$ is in the first $\nicefrac13^{rd}$ part of $\sigma(G - F_{opt})$. Such a vertex $x$ has in-degree and out-degree at least $\frac{n}{18\alpha}+\frac1{4\alpha}-\frac12$ (by \Cref{clm:degboundgoodver}). Consequently, both $|R_{G-P}(x)|$ and $|\overline{R}_{G - P}(x)|$ are at most $n(1-\nicefrac{1}{30\alpha})$ (by \Cref{eq:bounding-problem-size}).
    

We perform the following iterative procedure to compute $P$, which is initialized as $\emptyset$. Let $C$ be a shortest cycle in $G-P$ that contains $x$ and let $C' \subseteq C$ be the shortest induced cycle in $C$. 
As $x\notin F_{opt}$, therefore $F_{opt}\cap (C' \setminus\{x\}) \ne \emptyset$.
The crucial point to note here is that $|C'\setminus \{x\}|\le 2\alpha $ regardless of $|C|$ (by \Cref{obs:shortinducedcycles}).

Since $F_{opt}$ is an $FVS$ and $x \notin F_{opt}$, we have $|F_{opt}\cap (C^\prime\setminus\{x\})| \ge 1$. 
Even though $C'$ may not contain $x$, \textit{hitting} (i.e., picking vertices from) the cycle $C^\prime$ implies hitting the cycle $C$ that contains $x$. Now observe that if we are in the unweighted setup of the \DFVSbIN problem, then to hit the cycle $C'$ we can pick all vertices in $C' \setminus \{x\}$ in $P$.
Therefore, we are getting a $2\alpha$-approximate solution conditioning on the event that $x \notin F_{opt}$.

Observe that this strategy fails if we are in the weighted setup of the \DFVSbIN problem. We cannot simply pick all the vertices of $C'\setminus \{x\}$. We resolve this issue by using the \textit{``local ratio"} technique as follows.
We find the lightest vertex (say $u$) in $C^\prime \setminus \{x\}$, add it to $P$ and update the weights of each vertex $v \in C^\prime\setminus \{x\}$ to $w(v)-w(u)$. We repeat the procedure until there is no cycle in $G-P$ which contains vertex $x$. 

Let $w^\prime$ be the weight function at the end of the procedure.
Next we recursively get $2\alpha$-approximate FVSs in $(G[R_{G - P}(x)],w^\prime)$ (say $F_1$) and in $(G[\overline{R}_{G - P}(x)],w^\prime)$ (say $F_2$) resp., each with probability at least $\nicefrac12$. Thereafter we construct $P\cup F_1\cup F_2$ which is a $2\alpha$-approximate FVS in $(G,w)$ (by \Cref{clm:smalloptapproxfactor}) with probability at least $\nicefrac1{12}-\nicefrac1{18\alpha}$ (by \Cref{eqn:probability_lb}).

To boost the success probability of the algorithm to the required lower bound $\nicefrac{1}{2}$, we repeat the random experiment $28\alpha$ times, i.e. we repeat the procedure of sampling the vertex $x$, computing the set $P$, and solving the two recursive subproblems $(G[R_{G - P}(x)],w^\prime)$ and $(G[\overline{R}_{G - P}(x)],w^\prime)$, $28\alpha$ times.
\end{itemize}
But we don't know $F_{opt}$ during the execution of the algorithm (and hence which of the cases we fall into). Thus, we compute $28\alpha +1$ solutions for the instance $(G,w)$: one solution for the Case 1 and $28\alpha$ solutions for the Case 2. Thereafter we take the lightest among all the $28\alpha+1$ solutions. 

Now to analyse the time complexity of the algorithm, observe that in the case that $|F_{opt}|< \nicefrac{2n}{3\alpha}$, we are making two recursive calls on 
subproblems of size at most $n\cdot(1-\nicefrac{1}{30\alpha})$ and repeating the procedure for $28\alpha$ times; for the case when $|F_{opt}|\geq \nicefrac{2n}{3\alpha}$, we are making one recursive call to the instance $(G-L,w^\prime)$ of size $n\cdot (1- \nicefrac{1}{6\alpha})$. 
Thus, we get the recurrence $T(n)\le 2 \cdot 28 \alpha T\left(n\left(1-\frac{1}{30\alpha}\right)\right) + \bigoh(n^5)+T\left(n\left(1-\frac{1}{6\alpha}\right)\right)$ for the time complexity of the algorithm FindFVS, which solves to $n^{\bigoh(\alpha^2)}$.



\subsection{The Algorithm}

We compute $(28\alpha + 1)$ FVSs $\{F_i\}_{i=0}^{28\alpha}$ and return the lightest set among them. The algorithm is recursive. Each recursive call is made on a graph with strictly fewer vertices. When $|V(G)| \le 30 \alpha$ we solve the problem by brute force searching over all subsets of vertices.

\begin{definition}\label{def:update}We use two ``weight update" functions both of which take as input a weight function $w$ and a subset of vertices $Q$ and return a new weight function defined as follows.

\begin{enumerate}
    \item \textsf{update1}$(w,Q)$: Let $h$ be the heaviest vertex in $Q$. It returns $w': V(G) \setminus Q \rightarrow \mathbb{N}$ where 
    \[
    w'(v) = w(v) - w(h) \textnormal{ for each }v \in V(G) \setminus Q.
    \]
    \item \textsf{update2}$(w,Q)$: Let $\ell$ be the lightest vertex in $Q$. It returns $w': V(G) \rightarrow \mathbb{N}$ where 
    \[
    w'(v) = 
    \begin{cases}
        w(v)-w(\ell) & \textnormal{if } v \in Q;\\
        w(v) & \textnormal{otherwise}.
    \end{cases}
    \]
\end{enumerate}
\end{definition}
\begin{algorithm}[H]
  \scriptsize
  \begin{algorithmic}[1]    \caption{\textsf{FindFVS}}\label{alg:findfvs}
    \Require{a digraph $G=(V,E)$, vertex weights $w: V \rightarrow \mathbb{N}_{\ge 0}$}
    \Ensure{a subset of vertices $X$}
    \If{$n \le 30 \alpha$}
      \State Iterate over all subsets of $V(G)$, and return a optimal FVS of $G$, say $X$
    \EndIf
    \State Let $L$ be a set of $\nicefrac{n}{6\alpha}$ lightest vertices in $V(G)$ acc. to $w$
    \State $F_{0} := L \cup \textsf{FindFVS}\left(G - L, \textsf{update1}(w,L)\right)$
    \ForAll{$i \in [28\alpha]$}
      \State Pick $x_i$ uniformly at random from $V(G)$
      \If{$\min\left\{d^+_G(x_i), d^-_G(x_i)\right\} < \frac{n}{18\alpha}+\frac1{4\alpha}-\frac12$}
        \State $F_i:= V(G)$
        \State \textbf{continue}
      \EndIf
      \State $C_i = \{\}$
      \State $w_i = w$
      \While{there is a cycle in $G \!-\! C_i$ containing $x_i$} \Comment{\textcolor{blue}{Eliminating cycles with $x_i$}}
        \State Let $C$ be a shortest cycle in $G - C_i$ containing $x_i$
        \State Let $C' \subseteq C$ be a shortest induced cycle inside $C$    
        \State Let $v$ be a lightest vertex in $C' \setminus \{x_i\}$ acc. to $w_i$
        \State $C_i = C_i \cup \{v\}$    
        \State $w_i = \textsf{update2}(w_i, C' \setminus \{x_i\})$
      \EndWhile
      \State $F_i$ := $C_i \cup$ \textsf{FindFVS}($G[R_{G - C_i}(x_i)], w_i$) $\cup$ \textsf{FindFVS}($G[\overline{R}_{G - C_i}(x_i)], w_i$)
    \EndFor
    \State Return a lightest set among $F_0,F_1, \dots, F_{28 \alpha}$ acc. to $w$, say $X$
  \end{algorithmic}
\end{algorithm}

\subsubsection{Analysis}\label{ss:Analysis}
~\\
\noindent\textit{Proof of \Cref{thm:fvsbin}.}
We will prove this by induction on $n$. For the base case we consider $n \le 30 \alpha$; where by iterating over all subsets of vertices, we can, in $\bigoh(2^{30\alpha} n^2)$ time, find a minimum weight FVS in $G$. Hence, from now on, we will analyze when $n > 30 \alpha$. 



Let $F_{opt}$ be an optimal \FVS in $G$. If $|F_{opt}| \ge \nicefrac{2n}{3\alpha}$, then we claim that $F_0$ satisfies the theorem statement.
Note that the algorithm returns $X$ such that $w(X) \le w(F_i)$ for each $i \in [0,28\alpha]$, i.e. $X$ is the lightest set among $\{F_0,\ldots,F_{28\alpha}\}$.

\begin{newclaim}[$\dagger$]
\label[newclaim]{clm:largeopt}
Suppose that $|F_{opt}| \ge \nicefrac{2n}{3\alpha}$. Then, with probability at least $\nicefrac12$, $F_0$ is a $2\alpha$-approximate FVS in $(G,w)$.
\end{newclaim}

\begin{proof}
Let $w'$ denote the weight function returned by $\textsf{update1}(w,L)$, \Cref{def:update}. Let $F'$ denote the set returned by the recursive call $\textsf{FindFVS}(G - L, w')$. Let $v$ denote the heaviest vertex in $L$. By applying the induction hypothesis on $G -L$, we have that with probability at least $\nicefrac12$, $F'$ is a $2\alpha$-approximate FVS in $(G -L,w')$. Therefore, we have $w(F_0) = w(F' \cup L) \le 2 \alpha\cdot w(F_{opt})$.
\qed 
\end{proof}





Thus, from now on we assume that $|F_{opt}| < \nicefrac{2n}{3\alpha}$. 
Next, we will analyze the probabilistic events in our algorithm, by which we will obtain a lower bound on the 
algorithm accuracy.

Consider the ordering $\sigma(G - F_{opt}) = \langle v_1, \dots, v_{n-|F_{opt}|} \rangle$ of vertices in $G - F_{opt}$. For each $i \in [28\alpha]$, we say that the randomly chosen vertex $x_i$ is \textit{good} if $x_i \not \in F_{opt}$ and the position of $x_i$ in $\sigma(G-F_{opt})$ is 
in the first $\nicefrac{(n-|F_{opt}|)}{3}$ vertices. Let $E_i^1$ denote the event that $x_i$ is good. Thus, for each $i$, $E_i^1$ occurs with probability at least $\left(\frac{(n-|F_{opt}|)}{3}\right)/n \ge \frac13-\frac{2}{9\alpha}$. The underlying goal of this definition is to bound the size of the recursive subproblems. That is, if $x_i$ is good for some $i\in [28\alpha]$, then the size of the subproblem in the $i^{th}$ iteration of the for loop is 
bounded, 
established via \cref{eq:bounding-problem-size}. Towards this, we first show the following.


\begin{newclaim}[$\dagger$]
\label[newclaim]{clm:degboundgoodver}
If $x_i$ is good, then both $d^+_G(x_i)$ and $d^-_G(x_i)$ are at least $\frac{n}{18\alpha}+\frac1{4\alpha}-\frac12$. 
\end{newclaim}

Thus, if either $d^+_G(x_i)$ or $d^-_G(x_i)$ is strictly less than $\frac{n}{18\alpha}+\frac1{4\alpha}-\frac12$, we conclude that $x_i$ is not good. Then we set $F_i=V(G)$ and continue to find the next FVS (lines 8-11).

\noindent\textbf{Approximation factor analysis.} For a fixed $i \in [28 \alpha]$, suppose that $x_i$ is good (i.e., we condition on the event $E_i^1$). Let $G_1$ and $G_2$ denote $G[R_{G - C_i}(x_i)]$ and $G[\overline{R}_{G - C_i}(x_i)]$, resp. Let $F_i^+$ an $F_i^-$ denote the FVSs returned by the recursive calls \textsf{FindFVS}($G[R_{G - C_i}(x_i)], w_i$) and \textsf{FindFVS}($G[\overline{R}_{G - C_i}(x_i)], w_i$), resp.


Observe that all cycles in $G - C_i$ are contained completely inside either $G_1$ or $G_2$. Thus, $F_i^+ \cup F_i^-$ is an FVS in $G-C_i$ and $F_i = C_i \cup F^+_i \cup F^-_i$ is an FVS in $G$. Consequently, each set in $\{F_i\}_{i=0}^{28\alpha}$ is an FVS in $G$ and the algorithm always returns an FVS in $G$. Moreover, $F_{opt} \cap V(G_1)$ and $F_{opt} \cap V(G_2)$ are FVSs in $G_1$ and $G_2$ resp., by \Cref{lem:divide-conquer-FVS}. Thus, $F_{opt} \setminus C_i = F_{opt} \cap (V(G_1) \cup V(G_2))$ is an FVS in $G_1 \cup G_2 = G - C_i$, by \Cref{obs1}.

Let $E_i^2$ and $E_i^3$ denote the events that $F_i^+$ and $F_i^-$ are $2\alpha$-approximate FVS in $(G_1, w_i)$ and $(G_2,w_i)$, resp. By applying the induction hypothesis on $G_1$ and $G_2$, we have that each of $E_i^2$ and $E_i^3$ happens individually
with probability at least $\nicefrac12$. Since $E_i^2$ and $E_i^3$ are independent, both $E_i^2$ and $E_i^3$ happen with probability at least $\nicefrac12 \cdot \nicefrac12 = \nicefrac14$.

Suppose that $F_i^+$ and $F_i^-$ are $2\alpha$-approximate FVSs in $(G_1, w_i)$ and $(G_2,w_i)$, resp. Therefore, $F^+_i \cup F^-_i$ is a $2\alpha$-approximate FVS in $(G_1 \cup G_2,w_i)$. Consequently, we have that $w_i(F^+_i \cup F^-_i) \le 2 \alpha\cdot w_i(F_{opt} \setminus C_i)$,
since $F_{opt} \setminus C_i$ is an FVS in $G - C_i$, by \Cref{obs1}. From now on we condition on the events $E_i^1$, $E_i^2$, and $E_i^3$ and then prove the following. 

\begin{newclaim}[$\dagger$]
\label[newclaim]{clm:smalloptapproxfactor}
The set $F_i = C_i \cup F_i^+ \cup F_i^-$ is a $2\alpha$-approximate \FVS in $(G,w)$.
\end{newclaim}
\begin{proof}
Since 
$F_i$ is an FVS in $G$, it suffices to show that $w(C_i \cup F^+_i \cup F^-_i) \le 2\alpha\cdot w(F_{opt})$. 

Suppose that $C_i$ contains $\ell$ vertices at the end of the while loop (lines 14-20). Then, $w_i$, which was initially $w$, was updated $\ell$ times using the method $\textsf{update2}$. 
Let $w_i^0=w$. For each $j \in [\ell]$, let $w_i^j$ be the function $w_i$ after $j$ updates and let $v_j$ denote the $j^{th}$ vertex added to $C_i$.

We will prove a more general condition, which implies the claim, that for each $j \in [\ell]$ 
\begin{equation}
w_i^{j-1}(F^+_i \cup F^-_i \cup\{v_j, \dots, v_{\ell}\} \cup \{v_1, \dots, v_{j-1}\})
\le 2 \alpha\cdot w_i^{j-1}(F_{opt}).
\label{iterative-weight-update}
\end{equation}
Note that for a fixed $i$ and $j=1$, we have $w_i^0=w$ and $C_i=\{v_1, \ldots, v_{\ell}\}$ and so the above condition yields $w(F_i^+ \cup F_i^- \cup C_i) \leq 2\alpha\cdot w(F_{opt})$. 
Observe that by the definition of \textsf{update2}, 
\begin{equation}
w_i^{j-1}(v_k)=0 \textnormal{ for all } k \in [j-1].
\label{eqn:weight_function}
\end{equation}
Then, \Cref{iterative-weight-update} is equivalent to the following: 
\[
w_i^{j-1}(F^+_i \cup F^-_i \cup\{v_j, \dots, v_{\ell}\}) \le 2 \alpha\cdot w_i^{j-1}(F_{opt}\setminus\{v_1, \dots, v_{j-1}\}). \tag{by \cref{eqn:weight_function}}
\]

Our proof will use induction on the value of $j$, in decreasing order. 
For the base case $j=\ell+1$ (in which case $w_i^{j-1}=w_i^{\ell}=w_i$), we have 
\begin{align*}
w_i^{\ell}(F^+_i \cup F^-_i  \cup \{v_1, \dots, v_{\ell}\}) = w_i(F^{+}_{i} \cup F^-_i) &\le 2 \alpha\cdot w^{\ell}_i(F_{opt} \setminus C_i) \tag{by \cref{eqn:weight_function}}.
\end{align*}
We provide a proof of the inductive case in the full version of the paper.
This concludes the proof of the claim. \qed
\end{proof}

We will conclude the proof of the theorem by showing that our algorithm succeeds with bounded probability within time $\bigoh(n^{122\alpha^2})$. 


\noindent\textbf{Probability analysis.}
We have conditioned upon three events: \begin{enumerate*} \item[$(E_i^1)$] $x_i$ is good, \item[$(E_i^2)$] $F^+_i$ is a $2\alpha$-approximate FVS in $(G_1,w_i)$, and \item[$(E_i^3)$] $F^-_i$ is a $2\alpha$-approximate FVS in $(G_2,w_i)$. \end{enumerate*} For a fixed $i$, these three events happen with probability at least \begin{equation}
\label{eqn:probability_lb}
(\nicefrac13 - \nicefrac2{9\alpha})\cdot\nicefrac12\cdot\nicefrac12 = \nicefrac1{12} - \nicefrac1{18\alpha}.   
\end{equation}
The probability that  for each $i \in [28 \alpha]$ at least one of $\{E_i^1,E_i^2,E_i^3\}$ does not happen is at most $(\nicefrac{11}{12}+\nicefrac1{18\alpha})^{28 \alpha} \le \nicefrac12$ because $\alpha \ge 1$. Thus, with probability at least $\nicefrac12$ there exists $i \in [28\alpha]$ such that all the three events occur and consequently $F^+_i \cup F^-_i \cup C_i$ is a $2\alpha$-approximate FVS in $(G,w)$.


\noindent\textbf{Running time analysis.}
If $n \le 30 \alpha$, then the algorithm runs in $\bigoh(2^{30 \alpha}n^2)$ time. 
From now on, consider the case $n > 30 \alpha$. Each iteration of the while loop (lines 14-20) can be done in $\bigoh(n^3)$ time since finding a shortest cycle $C$ (line 15), a shortest induced cycle $C' \subseteq C$ (line 16), and a lightest vertex in $C'$ (line 17) can all be done in $\bigoh(n^3)$ time. 



Since in each iteration, a vertex $v \in G \setminus C_i$ is added to $C_i$ (which was initially empty), the repeat loop is carried out for at most $n$ steps. Therefore, the repeat loop can be done in time $\bigoh(n^4)$. Before we consider recursive calls, we would like to note that finding the $\nicefrac{n}{6\alpha}$ lightest vertices in $G$ (line 4) can be done in $\bigoh(n \log n)$ time and finding a lightest set (line 23) can be done in $\bigoh(n\alpha)$ time.

If $x_i$ is not good, then $F_i$ is set to $V(G)$ and no further recursive calls are made. Recall that by \Cref{clm:degboundgoodver}, if $x_i$ is good, both $d^+_G(x_i)$ and $d^-_G(x_i)$ are at least $\nicefrac{n}{18\alpha}+\nicefrac1{4\alpha}-\nicefrac12 $. 
Since $G - C_i$ does not contain any cycle that $x_i$ is part of, the number of vertices in $G[R_{G - C_i}(x_i))]$ and $G[\overline{R}_{G - C_i}(x_i)]$ is at most $n-(\nicefrac{n}{18\alpha}+\nicefrac1{4\alpha}-\nicefrac12)$. Upon simplification, we note that
\begin{align}
\label{eq:bounding-problem-size}
n-(\nicefrac{n}{18\alpha}+\nicefrac1{4\alpha}-\nicefrac12) \le n(1-\nicefrac{1}{18\alpha}) +\nicefrac12 \le n(1-\nicefrac{1}{30\alpha}),
\end{align}
where the last inequality follows from the assumption that $n>30\alpha$. 

Thus, the overall running time is given by an application of the Master theorem \cite{cormen2022introduction} to the recurrence relation
\begin{align}
\label{eq:fvsbin_rr}
T(n) \le 2 \cdot 28 \alpha T(n(1-\nicefrac{1}{30\alpha})) + 28\alpha\bigoh(n^4)+T(n(1-\nicefrac{1}{6\alpha})) = \bigoh(n^{122\alpha^2}).
\end{align}
This concludes the proof of the theorem.
\section{Subset FVS in Tournaments}



In addition to a tournament $T$ on $n$ vertices and a weight function $w: V(T) \rightarrow \mathbb{N}_{\geq 0}$, we are given as input a vertex subset $S \subseteq V(T)$ of size $s$. We say that $F\subseteq V(T)$ is a \textit{subset feedback vertex set} (\SFVS, in short) in $T$ if there is no cycle containing vertices of $S$ in $T - F$. The goal is to find a minimum weight \SFVS in $T$. Observe that if $S=V(T)$, then the problem is a case of \DFVSbIN with $\alpha=1$. 






By $(T,S,w)$, we denote an instance of \sfvst.
The following observations follow from the hereditary property of subset-acyclicity.

\begin{observation}\label{obs3}
    Let $F$ be an \SFVS in $(T,S,w)$ and let $X\subseteq V(G)$. Then $F\setminus X$ is an \SFVS in $(T-X,S\setminus X,w)$.
\end{observation}

\begin{observation}\label{obs4}
Suppose that $F$ is an optimal \SFVS in $(T,S,w)$ and $X$ is a subset of $F$. Then, $F \setminus X$ is an optimal \SFVS in $(T - X,S\setminus X,w)$, of weight $w(F)-w(X)$.
\end{observation}

In our discussions, a \textit{triangle} ($\triangle$, in short) is a directed cycle of length three. 
The following structural lemma gives us the fact that, hitting all triangles passing through the vertices of $S$ is equivalent to hitting all cycles passing through $S$. 
As a consequence, we have that $F \subseteq V(T)$ is an \SFVS if and only if in $T-F$ there is no triangle that contains a vertex of $S$.

\begin{lemma}[$\dagger$]
\label{prop:shortcycleinsfvst}
For a vertex $x\in S$, any shortest cycle containing $x$ is a $\triangle$.
\end{lemma}

The following lemma, analogous to \Cref{lem:divide-conquer-FVS}, shows the interaction of an \SFVS $F$ in $(T,S,w)$ with $T[R_{T}(x)]$ and $T[\overline{R}_T(x)]$. 



\begin{lemma}[$\dagger$]
\label{lem:divide-conquer-SFVS}Suppose that $x \in S$ is a vertex that is not part of any cycle in $T$, then the following holds: $F$ is an \SFVS in $(T,S,w)$ if and only if $F\cap R_{T}(x)$ is an \SFVS in $(T[R_{T}(x)],S\cap R_{T}(x),w)$ and $F\cap \overline{R}_T(x)$ is an \SFVS in $(T[\overline{R}_T(x)],S\cap \overline{R}_T(x),w)$.
\end{lemma}







Next, we present a randomized $2$-approximation algorithm for \sfvst that runs in time $n^{\bigoh(1)}$ extending the ideas that we used to solve \DFVSbIN. 

The base case of the recursive algorithm is given by $s=|S|\le 30$. Unlike Algorithm \ref{alg:findfvs}, we cannot handle the base case by simply iterating over all subsets of vertices of size at most $30$ to find an optimal solution. An \SFVS may contain vertices outside $S$ (i.e., in $T-S$). To overcome this, we use the notion of vertex covers.
\hide{\ma{This definition comes out of nowhere! and then suddenly we revert to talking about base case.}}
A subset of vertices $B \subseteq V(T)$ is called a \textit{vertex cover} in $T$ if for each arc $(u,v)\in E(T)$ we have $B \cap \{u,v\} \ne \emptyset$. Towards handling the base case, we use the well known fact that a $2$-approximate minimum weight vertex cover in a digraph on $n$ vertices can be computed in $\bigoh(n^{2})$ time \cite{bar2004local} using the subroutine $\textsf{FindVertexCover}(T',w')$ which takes $T'=(V',E')$ as input together with a weight function $w:V(T')\rightarrow \mathbb{N}_{\geq 0}$.
From now on, by $\binom{S}{\le30}$, we denote the subsets of $S$ of size at most 30. 

\subsection{Technical Overview}
\begin{definition}\label{def:update_sfvst}
Given a set $S\subseteq V(T)$, we use two ``weight update" functions both of which take as input a weight function $w$ and a set of vertices $Q$ and return a new weight function:
\begin{enumerate}
    \item \textsf{update3}$(w,Q)$: Let $h$ be the heaviest vertex in $Q$. It returns $w': V(T) \setminus Q \rightarrow \mathbb{N}_{\geq 0}$ where 
    \[
    w'(v) = 
    \begin{cases}
        w(v)-w(h) & \textnormal{if } v \in S \setminus Q;\\
        w(v) & \textnormal{otherwise}.
    \end{cases}
    \]
    \item \textsf{update4}$(w,Q)$: Let $\ell$ be the lightest vertex in $Q$. It returns $w': V(G) \rightarrow \mathbb{N}_{\geq 0}$ where 
    \[
    w'(v) = 
    \begin{cases}
        w(v)-w(\ell) & \textnormal{if } v \in Q;\\
        w(v) & \textnormal{otherwise}.
    \end{cases}
    \]
\end{enumerate}
\end{definition}
The high level structure of the solution is similar to the one for \DFVSbIN. In this overview, we highlight the key differences from the previous algorithm. Let $F_{opt}$ be an optimal SFVS in $T$. We consider the following three cases.
\begin{itemize}[wide=0pt]
    \item{\textbf{Case 1 ($|S \cap F_{opt}| \le 30$):}} We ``guess" this intersection by iterating over all $Q \in \binom{S}{\le30}$, i.e., we guess the part of $S$ which is \textit{inside} (say $Q$) the solution and the part which is \textit{outside} (say $O$) the solution. After the guessing if we find any $\triangle$ containing only vertices from $O$, then we cannot extend $Q$ and trivially set $S$ to be the solution, denoted by $F_Q$. Otherwise, we initially set $F_Q$ to be the vertices which are inside the solution and extend it in two phases. In the first phase, we deal with $\triangle$s where two of its vertices are in $O$ and we add the third vertex of such a $\triangle$ to $F_Q$. In the second phase, we deal with $\triangle$s where one vertex is in $O$: we find a 2-approximate weighted vertex cover using $\textsf{FindVertexCover}()$ on the (undirected) graph containing the edges between the end vertices that are not in $O$, of such $\triangle$s.
    We show that the extension $F_Q$ corresponding to $Q = S \cap F_{opt}$ is a $2$-approximate solution (\Cref{clm:basecasesfvst}).
    As a base case of the recursive algorithm, if $s<30$ then we return the lightest solution in $\{F_Q\}_{Q \in \binom{S}{\le30}}$ (say $Y$).
    \item{\textbf{Case 2 ($|S \cap F_{opt}|\ge \nicefrac{2s}{3}$):}}
    Let $L$ be a set of $\nicefrac{s}{6}$ lightest vertices in $S$ and $F$ be the SFVS returned by the recursive call $\textsf{FindSFVS}(T-L, S\setminus L, \texttt{update3}(w,L))$. We show that $F_0=F \cup L$ is a 2-approximate SFVS in $(T,S,w)$ with probability at least $\nicefrac12$ (\Cref{clm:largeoptsfvst}).
    \item{\textbf{Case 3 ($30 < |S \cap F_{opt}| < \nicefrac{2s}{3}$):}} First, we randomly sample a vertex $x \in S$. With probability at least $\nicefrac13-\nicefrac29$, $x$ is in the first $\nicefrac13^{rd}$ part of $\sigma(T[S\setminus F_{opt}])$. Such a vertex $x$ has in-degree and out-degree at least $\nicefrac{s}{18}+\nicefrac14-\nicefrac12$ in $T[S\setminus F_{opt}]$ (\Cref{clm:degboundgoodversfvst}). We compute $P$ (and a weight function $w'$), a set of vertices such that in $T-P$ there is no cycle that contains $x$. Consequently, both $|S \cap R_{T-P}(x)|$ and $|S \cap \overline{R}_{T-P}|$ are at most $s(1-\nicefrac1{30})$. We recursively compute SFVSs $F_1$ and $F_2$ in $(T[R_{T-P}(x)], S \cap R_{T-P}(x), w')$ and $(T[\overline{R}_{T-P}], S \cap \overline{R}_{T-P}, w')$ resp. and obtain an 2-approximate SFVS in $T$: $P \cup F_1 \cup F_2$ (\Cref{clm:smalloptsfvst}). Let $F_1, \dots, F_{28}$ be the solutions that we get by repeating the random experiment 28 times.

    \begin{sloppypar}
    We show that with probability at least $\nicefrac12$, the lightest set among $\{Y, F_0,F_1, \dots, F_{28}\}$ according to $w$ is a 2-approximate SFVS in $T$.
    \end{sloppypar}
\end{itemize}
For the running time, we show that the number of subproblems is $s^{\bigoh(1)}$ and the time spent at each subproblem is $n^{\bigoh(1)}$. Thus, the overall running time is $n^{\bigoh(1)}$.
\subsection{The Algorithm}

The full version of the paper contains the pseudocode of the algorithm for SFVS in tournaments. 
\subsubsection{Analysis}
~\\
\noindent\textit{Proof of \Cref{thm:sfvst}.}
Let $F_{opt}$ denote an optimal SFVS in $T$. We will prove by induction on $n$. For the base case, we consider $|S\cap F_{opt}| \le 30$ (which subsumes the case $n \le 30$). 
\begin{newclaim}[$\dagger$]
\label[newclaim]{clm:basecasesfvst}
If $|S\cap F_{opt}| \le 30$, then 
$Y$ is a $2$-approximate SFVS in $(T,S,w)$.
\end{newclaim}

From now on, we assume that $|S\cap F_{opt}|>30$. 
We will restate the key statements, claims and definitions. Proofs of the following claims are similar to that of the claims that we have proved for the \DFVSbIN problem. The only part of problem which need to be argued is the running time analysis, as in the \sfvst problem the base case is non-trivial. 

Now similar to before, we first consider the case $|S\cap F_{opt}| \ge \nicefrac{2s}{3}$. 
\begin{newclaim}[$\dagger$]
\label[newclaim]{clm:largeoptsfvst}
Suppose that $|S \cap F_{opt}| \ge \nicefrac{2s}{3}$. Then, with probability at least $\nicefrac12$, $F_0$ is a $2$-approximate SFVS in $(T,S,w)$.
\end{newclaim}
From now on, we assume that $|S \cap F_{opt}| < \nicefrac{2s}{3}$.
Consider the ordering $\sigma(T[S \setminus F_{opt}]) = \langle v_1, \dots, v_{s-|S \cap F_{opt}|} \rangle$ of vertices in $S \setminus F_{opt}$.
For each $i \in [28 ]$, we say that the randomly chosen vertex $x_i \in S$ is \textit{good} if $x_i \not \in F_{opt}$ and the position of $x_i$ in $\sigma(T[S \setminus F_{opt}])$ is at most $\nicefrac{(s-|S \cap F_{opt}|)}{3}$. Let $E_i^1$ denote the event that $x_i$ is good. Thus, for each $i$, $E_i^1$ occurs with probability at least $\left(\frac{(s-|S \cap F_{opt}|)}{3}\right)/s \ge \frac13-\frac{2}{9}$. Analogous to \Cref{clm:degboundgoodver}, we have the following.

\begin{newclaim}[$\dagger$]
\label[newclaim]{clm:degboundgoodversfvst}
If $x_i$ is good, then $\min(d^+_{T[S]}(x_i),d^-_{T[S]}(x_i))\ge\frac{s}{18}+\frac1{4}-\frac12$. 
\end{newclaim}

For a fixed $i \in [28]$, suppose that $x_i$ is good (i.e., we condition on the event $E_i^1$). 
Similar to \Cref{alg:findfvs}, the set $C_i$ and function $w_i$ are computed as follows. 
\begin{algo}
\footnotesize
    \begin{algorithmic}
      \State $C_i = \{\}, w_i=w$
      \While{there is a $\triangle$ in $T - C_i$ that contains $x_i$}\Comment{\textcolor{blue}{Eliminating $\triangle$s with $x_i$}}
        \State Let $C'$ be a $\triangle$ in $T - C_i$ containing $x_i$ 
        \State Let $v$ be a lightest vertex in $C' \setminus \{x_i\}$ acc. to $w_i$
        \State $C_i = C_i \cup \{v\}$
        \State $w_i = \textsf{update4}(w_i, C' \setminus \{x_i\})$  
      \EndWhile
    \end{algorithmic}
\end{algo}

Let $R_i^1 = R_{T - C_i}(x_i )$, $R_i^2 = T[\overline{R}_{T - C_i}(x_i)]$, $T_1=T[R_i^1]$, and $T_2= T[R_i^2]$. 
Observe that all cycles passing through a vertex in $S$ in $T - C_i$ are contained completely inside either $T_1$ or $T_2$; otherwise, $x_i$ would form a triangle with an arc from such a cycle (by \Cref{prop:shortcycleinsfvst}). Hence, we can deduce that $F_{opt} \setminus C_i= F_{opt} \cap (V(T_1) \cup V(T_2))$. Moreover, since $T_1 \cup T_2 = T-C_i$, we have that $F_{opt} \setminus C_i$ is an \SFVS in $(T- C_i, S\setminus\!C_i, w)$ (by \Cref{obs3}). 

\hide{\il{Not using this anywhere}
Thus, (due to Lines~29--32) $x_i$ cannot be part of any cycle in $T-C_i$. Thus, $F_{opt} \cap V(T_1)$ and $F_{opt} \cap V(T_2)$ are optimal solutions in $(T_1, S\cap V(T_1), w)$ and $(T_2, S\cap V(T_2), w)$, respectively (by \Cref{scor}).}

Let $F_i^+$ an $F_i^-$ denote the solutions returned by the recursive calls \textsf{FindSFVS}($T[R_i^1, S \cap R_i^1, w_i$) and \textsf{FindSFVS}($T[R_i^2], S \cap R_i^2, w_i$), resp. Given that $R_i^1=R_{T-C_i}(x_i)$ and $R_i^2=\overline{R}_{T-C_i}(x_i)$, we note that $F_i^+ \cup F_i^{-}$ is an \SFVS in $(T-C_i, S\setminus C_i, w_i)$ (by \Cref{lem:divide-conquer-SFVS}). Thus, $F_i=C_i\cup F_i^+ \cup F_i^{-}$ is an \SFVS in $(T, S, w_i)$.


Consequently, each set in $\{F_i\}_{i=0}^{28}$ is an \SFVS in $T$. Since each set in \\$\{F_Q\}_{Q \in \binom{S}{\le30}}$ is also an \SFVS, the algorithm always returns an \SFVS in $T$. Next, we will analyze the quality of the solution $F_i$.

\hide{
\il{++++++++READ ABOVE++++++++++++++++++}
Let $F_i^+$ an $F_i^-$ denote the SFVSs returned by the recursive calls \textsf{FindSFVS}($T[R_i^1], S \cap R_i^1, w_i$) and \textsf{FindSFVS}($T[R_i^2], S \cap R_i^2, w_i$), resp. 
\il{State weight functions..}
Observe that all cycles passing through a vertex in $S$ in $T - C_i$ are contained completely inside either $T_1$ or $T_2$, otherwise $x_i$ would form a triangle with an arc from such a cycle (by \Cref{prop:shortcycleinsfvst}). \Ma{Thus, $F_{opt} \cap V(T_1)$ and $F_{opt} \cap V(T_2)$ are optimal solutions in $T_1$ and $T_2$ resp (by \Cref{scor}). Moreover, $F_{opt} \cap (V(T_1) \cup V(T_2)) = F_{opt} \setminus C_i$ is an \SFVS in $T_1 \cup T_2 = T - C_i$ (by \Cref{obs3}). Thus, $C_i \cup F^+_i \cup F^-_i$ is an \SFVS in $T$ (by \Cref{lem:divide-conquer-SFVS}).} Consequently, each set in $\{F_i\}_{i=0}^{28}$ is an \SFVS in $T$. Since each set in $\{F_Q\}_{Q \in \binom{S}{\le30}}$ is also an \SFVS, the algorithm always returns an \SFVS in $T$. 
\ma{In what follows, we will analyse...}
Next, we will analyze the quality of the solution $F_i$.
\il{+++++++++++++++++++++++++++++++++++}
}

Let $E_i^2$ and $E_i^3$ denote the events that $F_i^+$ and $F_i^-$ are $2$-approximate SFVSs in $(T_1, S \cap R_i^1, w_i)$ and $(T_2, S \cap R_i^1, w_i)$, resp. From now on, we condition on the events $E_i^1$, $E_i^2$, and $E_i^3$.

By applying the induction hypothesis on $T_1$ and $T_2$, we have that each of $E_i^2$ and $E_i^3$ happens individually with probability at least $\nicefrac12$. As $F_i^+$ and $F_i^-$ are $2$-approximate \SFVS in $(T_1, S \cap R_i^1, w_i)$ and $(T_2, S \cap R_i^2, w_i)$, resp., $F_i^+ \cup F_i^-$ is a 2-approximate \SFVS in $(T-C_i, S\setminus C_i, w_i)$. Since, $F_{opt} \setminus C_i$ is also an \SFVS in $(T - C_i, S\setminus C_i, w_i)$, we can infer that $w_i(F^+_i \cup F^-_i) \le 2 \cdot w_i(F_{opt} \setminus C_i)$. The following result is analogous to \Cref{clm:smalloptapproxfactor}.

\begin{newclaim}[$\dagger$]
\label[newclaim]{clm:smalloptsfvst}
The set $F_i = C_i \cup F_i^+ \cup F_i^-$ is a $2$-approximate \SFVS in $(T,S,w)$.
\end{newclaim}

We will conclude the proof of the theorem by showing that our algorithm succeeds with bounded probability in time $n^{\bigoh(1)}$. 

\noindent\textbf{Probability analysis.} 
With probability at least $\nicefrac12$ there exists $i \in [28]$ such that all the three events occur and $F^+_i \cup F^-_i \cup C_i$ is a $2$-approximate SFVS in $(T,S,w)$.

\noindent\textbf{Running time analysis.} 
If $|S\cap F_{opt}| \le 30$, then since there are $\bigoh(n^{30})$ subsets of $S$ of size at most 30 and for each subset, 
its extension to a solution can be computed in time $n^{\bigoh(1)}$,
the set $\{F_Q\}_{Q \in \binom{S}{\le30}}$ can be computed in time $n^{\bigoh(1)}$.
Else for each $i$, the set $C_i$ and function $w_i$ can be computed in time $\bigoh(n^4)$. Finding the $\nicefrac{s}{6}$ lightest vertices in $S$ can be done in $\bigoh(s \log s)$ time.

If $x_i$ is not good, then $F_i$ is set to $S$ and no further recursive calls are made.
By \Cref{clm:degboundgoodversfvst}, if $x_i$ is good, both $d^+_{T[S]}(x_i)$ and $d^-_{T[S]}(x_i)$ are at least $\nicefrac{s}{18}+\nicefrac1{4}-\nicefrac12 $. Since $T - C_i$ does not contain any cycle that $x_i$ is part of, the number of vertices in $S \cap R_i^1$ and $S \cap R_i^2$ is at most $s-(\nicefrac{s}{18}+\nicefrac1{4}-\nicefrac12) \le s(1-\nicefrac{1}{18}) +\nicefrac12 \le s(1-\nicefrac{1}{30})$ (since we have assumed that $s>30$).
The total number of recursive subproblems is given by an application of the Master theorem \cite{cormen2022introduction} to the recurrence relation
$T(s) \le 2 \cdot 28 T(s(1-\nicefrac{1}{30})) + T(30) +T(s(1-\nicefrac{1}{6}))= \bigoh(s^{122})$.


Thus, the overall running time is $T(s) \cdot (\textnormal{time spent at each subproblem}) = T(s) \cdot n^{\bigoh(1)} = n^{\bigoh(1)}$ (since $s \in [n]$). 
This concludes the proof of the theorem.

\section*{Acknowledgement}
We thank the anonymous reviewers for their helpful comments and suggestions.

\bibliography{main_latin}

\appendix

\section*{Appendix}

\section{Extended Preliminaries}
In this paper, we deal with simple directed graphs (\textit{digraphs}, in short) containing no parallel edges/arcs.
By $V(G)$ and $E(G)$, we denote the vertices and arcs of a digraph $G$ resp. We work in the setting of vertex weighted digraphs:  by $(G,w)$ we denote a vertex weighted digraph $G$ with weight function $w: V(G)\rightarrow \mathbb{N}_{\ge 0}$. The weight of a subset of vertices is the sum of weights of the vertices in the subset. Note that the setting of unweighted graphs is a special case of weighted graphs.
If there is an arc $(u,v) \in E(G)$, then $u$ is an \textit{in-neighbor} of $v$, and $v$ is an \textit{out-neighbor} of $u$.  
The \textit{in-neighborhood} of a vertex $x$, denoted by $N^{-}(x)=\{v| (v,x)\in E(G)\}$, is the set of in-neighbors of $x$. 
The \textit{in-degree} of a vertex $x$, denoted by $d^{-}_{G}(x)=|N^{-}(x)|$, is the number of in-neighbors of $x$.
The \textit{out-neighborhood} and \textit{out-degree} of a vertex $x$, denoted by $N^{+}(x)$ and $d^{+}_{G}(x)$ resp., are defined analogously. \textit{Deleting} a
vertex $v$ from $G$ involves removing the vertex $v$ from $V(G)$ and all those arcs in $E(G)$ that is incident to $v$. For a subset of vertices $S\subseteq V(G)$, we use $G - S$ to denote the digraph obtained by deleting all vertices of $S$ from $G$.
For a subset of vertices $S\subseteq V(G)$, the subgraph of $G$ \textit{induced} by $S$, denoted by $G[S]$, is the digraph on vertex set $S$ whose arcs are given by the arcs in $G$ with both end-vertices in $S$. For any induced subgraph $H$ of a vertex weighted graph $(G,w)$, we will assume that $w$ defines a weight function when restricted to $V(H)$.
An \textit{independent set} is a subset of vertices in $G$ that induces a digraph with no arcs. 
The \textit{independence number} of $G$, denoted by $\alpha(G)$, is the size of a largest independent set in $G$; we write $\alpha(G)$ as $\alpha$ when it is clear from context.


A \textit{directed path} of length $k$ is a sequence of distinct vertices $\langle x_1, \dots, x_k \rangle$ such that for every $i \in [k-1]$ we have $(x_{i}, x_{i+1}) \in E(G)$. We say a vertex $u$ is \textit{reachable} from a vertex $v$, if there is a directed path which contains both the vertices $u$ and $v$ and vertex $v$ appears before vertex $u$ in the sequence of the vertices which defines the directed path. 
By $R_{G}(x)$, we denote the set of all vertices other than $x$ reachable from $x$ in $G$. 
By $\overline{R_{G}}(x)$, we denote the set of all vertices not reachable from $x$ in $G$.
Observe that $(R_{G}(x), \overline{R_{G}}(x), \{x\})$ form a partition of $V(G)$.
The digraph $G$ is \textit{strongly connected} if for $\forall x,y\in V(G)$ where $x\neq y$ there is a directed path from $x$ to $y$.
A \textit{directed cycle} of length $k$ is a sequence of distinct vertices $\langle x_1, \dots, x_k \rangle$ such that for every $i \in [k-1]$ we have $(x_{i}, x_{i+1}) \in E(G)$, and $(x_k, x_1) \in E(G)$. 
A digraph is \textit{acyclic} if it does not contain a directed cycle. 

A \textit{feedback vertex set} (FVS) in $G$ is a subset of vertices $S\subseteq V(G)$ such that $G-S$ is acyclic. 
Given a family $X=\{S_1, \ldots , S_l\}$ where $S_i\subseteq V(G)$ for each $i\in [l]$, we call $S_a\in X$ \textit{lightest} if for all $i\in [l]$ we have $w(S_a)\leq w(S_i)$. Similarly, we call $S_b\in X$ \textit{heaviest} if for all $i\in [l]$ we have $w(S_i)\leq w(S_b)$. An FVS $F_{opt}$ in $G$ is an \textit{optimal} solution of the instance $(G,w)$ if for every other FVS $S$ in $G$ we have $w(S)\geq w(F_{opt})$ i.e. $F_{opt}$ is lightest among all FVSs in $G$. An optimal FVS is often reffered as a \textit{minimum} FVS. An FVS $F$ in $G$ is called a \textit{$2\alpha$-approximate solution} of the instance $(G,w)$ if $w(F)\leq 2\alpha\cdot w(F_{opt})$.

An algorithm is a \textit{factor-$f$ randomized approximation algorithm} for a problem $\mathcal{P}$ if, for each instance $\mathcal{I}$ of $\mathcal{P}$, with probability at least $\nicefrac12$, it returns a solution for $\mathcal{I}$ of size at most $f\times OPT_{\mathcal{I}}$ where $OPT_{\mathcal{I}}$ denotes the optimal solution for the instance $\mathcal{I}$.




    





We define an \textit{HL-degree ordering} $\sigma(G) = \langle v_1, \dots, v_n \rangle$ on the vertices of $G$ by the recursive application of \Cref{clm:hiinhioutver}. Thus, $v_1$ has both $d^{+}_G(v_1)$ and $d^{-}_G(v_1)$ are at least $\nicefrac{(n-2 \alpha)}{4 \alpha}$; $v_2$ is a vertex such that both $d^{+}_{G - \{v_1\}}(v_2)$ and $d^{-}_{G - \{v_1\}}(v_2)$ are at least $\nicefrac{(n-1-2 \alpha)}{4 \alpha}$. More generally, for any $i \in \{2,\dots, n\}$, there are $(n-(i-1))$ vertices in $G -  \{v_1,\dots ,v_{i-1}\}$. Therefore for any $i\in \{2,\dots, n\}$ we have $d^{+}_{G -  \{v_1,\dots ,v_{i-1}\}}(v_i)$ and $d^{-}_{G -  \{v_1,\dots ,v_{i-1}\}}(v_i)$ are at least $\nicefrac{(n- (i-1) - 2 \alpha)}{4 \alpha}$.





\section{Missing Proofs}
\label{missing_proofs}

\subsection{Proof of \Cref{clm:hiinhioutver}}
    Consider the partition of $V(G)$ into two parts based on the in-degree and out-degree of vertices: ($V_{1}, V_{2}$) where $V_{1} = \{v \in V(G): d^{+}_G(v) \ge d^{-}_G(v)\}$ and $V_{2} = V(G) \setminus V_{1}$. At least one of the sets must have $\nicefrac{n}{2}$ vertices. 
    Without loss of generality, we may assume that $|V_{1}| \ge \nicefrac{n}{2}$. Therefore using \cref{lem: highinout}, we can conclude $\exists \hat{v}\in V_1$ such that $d^{-}_{G[V_1]}(\hat{v}) \ge \frac{|V_1| -\alpha}{2\alpha}= \frac{n-2\alpha}{4\alpha}$ as $G[V_{1}]$ has independence number $\alpha$. Combining this with the definition of $V_{1}$, we have that $d^{+}_G(\hat{v}) \ge d^{-}_G(\hat{v}) \ge d^{-}_{G[V_1]}(\hat{v}) \ge \nicefrac{(n - 2 \alpha)}{4 \alpha}$.
    


    


\subsection{Proof of \Cref{lem:divide-conquer-FVS}}
Forward direction is trivial as acyclicity is a hereditary property. 
For the backward direction, assume for the sake of contradiction that $F$ is not an \FVS in $G$: there is a cycle $C$ in $G - F$. Since $C$ is not a cycle in either $G[R_{G}(x)]$ or $G[\overline{R}_G(x)]$, there exists $x_1,x_2\in C$ such that $x_1 \in R_{G}(x)$ and $x_2 \in \overline{R}_G(x)$. But this is not possible because there is a path from $x_1$ to $x_2$ using arcs in $C$, a contradiction.

\subsection{Details in proof of \Cref{clm:largeopt}}

\begin{align*}
    w(F_0) &= w(F' \cup L)   =      w(F') + w(L)
                            =      w'(F') + |F'| w(v) +w(L)\\
                            &\le      2 \alpha\cdot  w'(F_{opt}\setminus L) + |F'| w(v) + w(L)
                            \tag{as $F_{opt} \setminus L$ is also an FVS in $G - L$}\\
                            &\le 2 \alpha\cdot (w(F_{opt})-|F_{opt} \setminus L| w(v)) + |F'| w(v) + w(L) \\
                            &\leq 2\alpha\cdot (w(F_{opt})-(\nicefrac{2n}{3\alpha}-\nicefrac{n}{6\alpha}) w(v))+ |F'| w(v) + w(L) \tag{since $|F_{opt}| \ge \nicefrac{2n}{3\alpha}$, $|L|=\nicefrac{n}{6\alpha}$}\\
                            &= 2\alpha\cdot (w(F_{opt})-(\nicefrac{n}{2\alpha}) w(v))+ |F'| w(v) + w(L)\\
                            &\le    2\alpha\cdot w(F_{opt}) - n\cdot w(v) + |F'| w(v) + |L| w(v)
                            \tag{since $v$ is the heaviest vertex in $L$}\\
                            &\le    2 \alpha\cdot w(F_{opt}) \tag{since $F'$ and $L$ are disjoint subsets of $V(G)$, i.e. $|F^\prime|+|L|\leq n$}.
\end{align*} 

Therefore, $F_0$ is a $2\alpha$-approximate FVS in $(G,w)$. 


\subsection{Proof of \Cref{clm:degboundgoodver}}
Let $T=G-F_{opt}$. 
Since $x_i$ is good, we have $x_i\in \{v_1,\dots,v_{|T|/3}\}$. 
Recall the degree bounds obtained by using \Cref{lem: highinout} in \Cref{subsec:our_methodology}. Applying the same on the subgraph $T$ (which also has independence number at most $\alpha$), we have
\begin{align*}
d^+_G(x_i) \ge d^{+}_{G - (F_{opt} \cup (v_1, \dots, v_{(|T|/3)-1})}(v_{|T|/3})
&\ge \frac{(n-(|F_{opt}|+\frac{|T|}{3} - 1)-2 \alpha)}{4\alpha} \tag{by \Cref{lem: highinout}}\\
=\frac{(n-(|F_{opt}|+\frac{n-|F_{opt}|}{3} - 1)-2 \alpha)}{4\alpha}
& \ge \frac{n}{6\alpha}-\frac{n}{9\alpha^2} +\frac1{4\alpha}-\frac12 \tag{since $|F_{opt}|< \nicefrac{2n}{3\alpha}$}\\
&\ge \frac{n}{18\alpha}+\frac1{4\alpha}-\frac12 \tag{since $\alpha \ge 1$}.\end{align*}
Similarly, we have $d^-_G(x_i) \ge \frac{n}{18\alpha}+\frac1{4\alpha}-\frac12$. 

\subsection{Details in proof of \Cref{clm:smalloptapproxfactor}}

\noindent\textbf{Inductive case:} For the inductive step, note that $v_j$ is a lightest vertex in $C'$ where $C' \subseteq C$ is a shortest induced cycle inside $C$. Clearly, $F_{opt}$ contains at least one vertex from $C'$, 
hence
$|F_{opt} \cap C'|\geq 1$. Moreover, if $x_i \notin C'$, by \Cref{obs:shortinducedcycles} we know that $|C'| \leq 2\alpha$, and so $1\le|F_{opt} \cap C'|\le 2\alpha$ follows. Else, we have $x_i\in C'$. But since $x_i \notin F_{opt}$ (it is good); 
hence,
$1 \le |F_{opt} \cap C'| = |F_{opt} \cap (C' \setminus \{x_i\})| \le 2\alpha$. Recall that $w_i^j = \textsf{update2}(w_i^{j-1}, C'\setminus\{x\})$. Thus, we have
\begin{align*}
&w_i^{j-1}(F^+_i \cup F^-_i \cup\{v_j, \dots, v_{\ell}\}\cup \{v_1, \dots, v_{j-1}\}) = w_i^{j-1}(F^+_i \cup F^-_i \cup \\
&\qquad \{v_j\} \!\cup \!\{v_{j+1}, \dots, v_{\ell}\})\tag{by \cref{eqn:weight_function}}\\
&  \le w_i^j(F^+_i \cup F^-_i \cup \{v_{j+1}, \dots, v_{\ell}\})+2\alpha \cdot w_i^{j-1}(v_j) \tag{since $|C' \setminus \{x_i\}| \le 2\alpha$}\\
& \le 2\alpha \cdot \left(w_i^j(F_{opt} \!\setminus \{v_{1}, \dots, v_j\})+ w_i^{j-1}(v_j)\right) \tag{by the induction hypothesis}\\
&  = 2\alpha \cdot\left(w_i^j(F_{opt} \!\setminus \{v_{1}, \dots, v_{j-1}\})+ w_i^{j-1}(v_j)\right) \tag{by \cref{eqn:weight_function}, $w_i^{j}(v_j)=0$} \\
& \le 2 \alpha \cdot \Big(w_i^{j-1}(F_{opt}\setminus \{v_{1}, \dots, v_{j-1}\})-\lvert(F_{opt}\setminus \{v_{1}, \dots, v_{j-1}\})\cap (C'\setminus \{x_i\})\rvert\\
& \qquad \cdot w_i^{j-1}(v_j)+w_i^{j-1}(v_j)\Big) 
\\
&\le 2\alpha \cdot w_i^{j-1}\left(F_{opt} \setminus \{v_1,\dots,v_{j-1}\}\right)\tag{since $|(F_{opt}\setminus \!\{v_{1}, \dots, v_{j-1}\}) \cap (C'\setminus \{x_i\})| \ge 1$}\\
&= 2\alpha \cdot w_i^{j-1}(F_{opt}).
\end{align*}


\subsection{Solution of \Cref{eq:fvsbin_rr}}
\begin{align*}
T(n) &\le 2 \cdot 28 \alpha T\left(n\left(1-\frac{1}{30\alpha}\right)\right) + 28\alpha\bigoh(n^4)+T\left(n\left(1-\frac{1}{6\alpha}\right)\right)\\
&\le 57\alpha T\left(n\left(1-\frac{1}{30\alpha}\right)\right) + \bigoh(n^5)\tag{since $n > 30\alpha$}\\
&= \bigoh(n^{\log_{t}{57\alpha}}) \tag{for $t= \frac{30\alpha}{30\alpha-1}$, $\log_{t}{57} >5$ for all $\alpha \ge 1$}\\
&= \bigoh(n^{122\alpha}).
\end{align*}

\subsection{Proof of \Cref{prop:shortcycleinsfvst}}

    Assume for contradiction that $C=\langle x,a_1,\ldots,a_r \rangle$ is a shortest cycle containing $x$, where $r\geq 3$. Clearly if $(a_2,x) \in E(T)$, then $C^\prime=\langle x,a_1,a_2 \rangle$ is a directed cycle of length 3 that contains $x$, a contradiction. As $T$ is a tournament, we must have $(x,a_2) \in E(T)$. Then, $C^{\prime\prime} = \langle x,a_2,\ldots a_r\rangle \subset C$ is a cycle containing $x$ which is shorter than $C$, a contradiction.

    
\subsection{Proof of \Cref{lem:divide-conquer-SFVS}}
Forward direction is trivial as acyclicity is a hereditary property. 
For the backward direction, assume for the sake of contradiction that $F$ is not an \SFVS in $(T,S,w)$. Let $C$ be a triangle containing a vertex of $S\setminus F$ in $T-F$. Since $C$ is not a cycle in either $T[R_{T}(x)]$ containing vertices of $S\cap R_{T}(x)$ or $T[\overline{R}_T(x)]$ containing vertices of $S\cap \overline{R}_T(x)$, there exists $x_1,x_2\in C$ such that $x_1 \in S\cap R_{G}(x)$ and $x_2 \in S\cap \overline{R}_G(x)$. But this is not possible because there is an arc from $x_1$ to $x_2$ (also in $C$), a contradiction.


\subsection{Proof of \Cref{clm:basecasesfvst}}

    Consider the execution of the for loop (lines 1-15) during which $Q = F_{opt} \cap S$. We claim that $w(F_Q) \le 2\cdot w(F_{opt})$.

    There does not exist $a,b,c \in O$ such that $abc$ form a $\triangle$ in $T$ (otherwise $F_{opt} \cap \{a,b,c\} \ne \emptyset$). Thus, the execution does not enter line 4. From the construction of $F_Q$, we have the following.\begin{observation}
    \label{obs:cardtwointersection}
    For each $a,b \in O, c\in V(T) \setminus (I \cup O)$ such that $abc$ form a $\triangle$ in $T_I$, both $F_{opt} \cap \{a,b,c\} = \{c\}$ and $F_Q \cap \{a,b,c\} = \{c\}$.
    \end{observation}

    Any triangle $abc$ with $\{a,b,c\} \cap S \ne \emptyset$ that we have not considered till now contains exactly one vertex from $O$. We may assume without loss of generality that $a \in O$. Then, $F_{opt}$ contains at least one of $\{b,c\}$ (otherwise $T - F_{opt}$ would contain the triangle $abc$). For each such triangle $abc$, we have the arc $(b,c)$ in the set $E'$. Let $\hat{C}$ denote the $2$-approximate vertex cover returned by \textsf{FindVertexCover}$((V(T),E'),w)$. Let $V'$ denote the set of endpoints of arcs in $E'$. Since $F_{opt}$ is an SFVS in $T$, it is a vertex cover in $((V(T),E'),w)$. Therefore, $w(\hat{C}) = w(F_Q \cap V') \le 2\cdot w(F_{opt} \cap V')$. 

    From the construction of $F_Q$, we have $w(F_Q \cap I) = w(F_{opt} \cap I)$. From \Cref{obs:cardtwointersection}, we have $w(F_Q \setminus (I \cup V')) = w(F_{opt} \setminus (I \cup V'))$. Combining all, we have
    \begin{align*}
        w(F_Q) &= w((F_Q \cap I) \cup (F_Q \cap V') \cup (F_Q \setminus (I \cup V')))\\
        &= w(F_Q \cap I) + w(F_Q \cap V') + w(F_Q \setminus (I \cup V'))\\
        &\le w(F_{opt} \cap I) + 2\cdot w(F_{opt} \cap V') + w(F_{opt} \setminus (I \cup V'))
        \le 2\cdot w(F_{opt}).
    \end{align*}
    This concludes the proof of the claim.


\subsection{Proof of \Cref{clm:largeoptsfvst}}
Let $w'$ denote the weight function returned by $\textsf{update3}(w,L)$ and $F'$ denote the set returned by the recursive call $\textsf{FindSFVS}(T - L, S\setminus L, w')$. Let $v$ be the heaviest vertex in $L$. By applying the induction hypothesis on $G -L$, we have that with probability at least $\nicefrac12$, $F'$ is a $2$-approximate \SFVS in $(G -L,w')$. Now we have the following, 
\begin{align*}
    w(F_0) &= w(F' \cup L)   =      w(F') + w(L) 
                            =      w'(F') + |F'\cap S| w(v) +w(L)\\
                            &\le      2\cdot  w'(F_{opt}\setminus L) + |F'\cap S| w(v) + w(L)
                            \tag{as $F_{opt} \setminus L$ is also an \SFVS in $G - L$}\\
                            &\le 2\cdot (w(F_{opt})-|(F_{opt} \setminus L) \cap S| w(v)) + |F'| w(v) + w(L) \\
                            &\leq 2\cdot (w(F_{opt})-(\nicefrac{2s}{3}-\nicefrac{s}{6}) w(v))+ |F'| w(v) + w(L) \tag{since $|F_{opt} \cap S| \ge \nicefrac{2s}{3}$, $|L|=\nicefrac{s}{6}$}\\
                            &\le 2\cdot (w(F_{opt})-(\nicefrac{s}{2}) w(v))+ |F'\cap S| w(v) + w(L)\\
                            &\le    2\cdot w(F_{opt}) - s\cdot w(v) + |F'\cap S| w(v) + |L| w(v)
                            \tag{since $v$ is the heaviest vertex in $L$}\\
                            &\le    2\cdot w(F_{opt}) \tag{since $F'\cap S$ and $L$ are disjoint subsets of $S$, i.e. $|F^\prime \cap S|+|L|\leq s$}.
\end{align*}
Therefore, $F_0$ is a $2$-approximate \SFVS in $(G,w)$.




\subsection{Proof of \Cref{clm:degboundgoodversfvst}}

Since $x_i$ is good, we have $x_i\in \{v_1,\dots,v_{|S \setminus F_{opt}|/3}\}$. 
Recall the degree bounds obtained by using \Cref{lem: highinout} in \Cref{subsec:our_methodology}. Applying the same on the tournament $T[S] - (S \cap F_{opt})$, we have

\begin{align*}
d^+_{T[S]}(x_i)) &\ge d^{+}_{T[S] - ((S \cap F_{opt}) \cup (v_1, \dots, v_{(s-|S \cap F_{opt}|)/3 -1}))}(v_{(s-|S \cap F_{opt}|)/3})\\
&\ge \frac{(s-(|S \cap F_{opt}|+\frac{s - |S \cap F_{opt}|}{3} - 1)-2 )}{4} \tag{by \Cref{lem: highinout}}\\
&\ge \frac{s}{6}-\frac{s}{9} +\frac1{4}-\frac12 \tag{since $|S \cap F_{opt}|< \nicefrac{2s}{3}$}\\
&\ge \frac{s}{18}+\frac1{4}-\frac12.\end{align*}
Similarly, we have $d^-_{T[S]}(x_i) \ge \frac{s}{18}+\frac1{4}-\frac12$. 

\subsection{Proof of \Cref{clm:smalloptsfvst}}

Since 
$F_i$ is an SFVS in $T$, it suffices to show that $w(C_i \cup F^+_i \cup F^-_i) \le 2\cdot w(F_{opt})$. 

Suppose that $C_i$ contains $\ell$ vertices at the end of the while loop. Then, $w_i$, which was initially $w$, was updated $\ell$ times using the method $\textsf{update4}$. 
Let $w_i^0=w$. For each $j \in [\ell]$, let $w_i^j$ denote the function $w_i$ after $j$ updates and let $v_j$ denote the $j^{th}$ vertex added to $C_i$.

We will prove a more general form of the claim that for each $j \in [\ell]$  we have 
\begin{equation}
w_i^{j-1}(F^+_i \cup F^-_i \cup\{v_j, \dots, v_{\ell}\} \cup \{v_1, \dots, v_{j-1}\})
\le 2\cdot w_i^{j-1}(F_{opt}).
\label{iterative-weight-update_s}
\end{equation}
Note that for a fixed $i$ and $j=1$, we have $w_i^0=w$ and $C_i=\{v_1, \ldots, v_{\ell}\}$ and so the above condition yields $w(F_i^+ \cup F_i^- \cup C_i) \leq 2\cdot w(F_{opt})$. 
Observe that by the definition of \textsf{update4}, 
\begin{equation}
w_i^{j-1}(v_k)=0 \textnormal{ for all } k \in [j-1].
\label{eqn:weight_function_s}
\end{equation}
Then, \Cref{iterative-weight-update_s} is equivalent to the following: 
\[
w_i^{j-1}(F^+_i \cup F^-_i \cup\{v_j, \dots, v_{\ell}\}) \le 2\cdot w_i^{j-1}(F_{opt}\setminus\{v_1, \dots, v_{j-1}\}). \tag{by \cref{eqn:weight_function_s}}
\]

Our proof will use induction on $j$. 

\noindent\smallskip\textbf{Base Case:} For the base case $j=\ell+1$ (in which case $w_i^{j-1}=w_i^{\ell}=w_i$), we have 
\begin{align*}
w_i^{\ell}(F^+_i \cup F^-_i  \cup \{v_1, \dots, v_{\ell}\}) = w_i(F^{+}_{i} \cup F^-_i) &\le 2\cdot w^{\ell}_i(F_{opt} \setminus C_i) \tag{by \cref{eqn:weight_function_s}}.
\end{align*}

\smallskip
\noindent\textbf{Inductive case:} For the inductive step, note that $v_j$ is a lightest vertex in $C'\setminus\{x_i\}$ where $C'$ is triangle in $T-C_i$ containing $x_i$. Clearly, $F_{opt}$ contains at least one vertex from $C'$, 
hence
$|F_{opt} \cap C'|\geq 1$. 
Since $x_i$ is good we have $x_i \notin F_{opt}$ and 
hence
$1 \le |F_{opt} \cap C'| = |F_{opt} \cap (C' \setminus \{x_i\})| \le 2$. Recall that $w_i^j = \textsf{update4}(w_i^{j-1}, C'\setminus\{x\})$. Thus, we have
    \begin{align*}
&w_i^{j-1}(F^+_i \cup F^-_i \cup\{v_j, \dots, v_{\ell}\}\cup \{v_1, \dots, v_{j-1}\}) \\
&= w_i^{j-1}(F^+_i \cup F^-_i \cup \{v_j\} \cup \{v_{j+1}, \dots, v_{\ell}\})\tag{by \cref{eqn:weight_function_s}}\\
&  \le w_i^j(F^+_i \cup F^-_i \cup \{v_{j+1}, \dots, v_{\ell}\})+2 \cdot w_i^{j-1}(v_j) \tag{since $|C' \setminus \{x_i\}| = 2$}\\
& \le 2\cdot \left(w_i^j(F_{opt} \!\setminus \{v_{1}, \dots, v_j\})+ w_i^{j-1}(v_j)\right) \tag{by the induction hypothesis}\\
&  = 2\cdot\left(w_i^j(F_{opt} \!\setminus \{v_{1}, \dots, v_{j-1}\})+ w_i^{j-1}(v_j)\right) \tag{by \cref{eqn:weight_function_s}}\\
& \le 2 \cdot \Big(w_i^{j-1}(F_{opt}\setminus \{v_{1}, \dots, v_{j-1}\})-\lvert(F_{opt}\setminus \{v_{1}, \dots, v_{j-1}\})\cap\\
& \qquad (C'\setminus \{x_i\})\rvert ~w_i^{j-1}(v_j)+ w_i^{j-1}(v_j)\Big) 
\\
&\le 2\cdot w_i^{j-1}\left(F_{opt} \setminus \{v_1,\dots,v_{j-1}\}\right)\tag{since $|(F_{opt} \setminus \{v_1,\dots,v_{j-1}\}) \cap (C'\setminus \{x_i\})| \ge 1$}\\
&= 2 \cdot w_i^{j-1}(F_{opt}).
\end{align*}
This concludes the proof of the claim.

\section{Algorithm for Subset FVS in Tournaments}
\label{algosfvst}
\begin{algorithm}[H]
  \begin{algorithmic}[1]
    \caption{\textsf{FindSFVS}}\label{alg:findsfvs}
    \Require{A tournament $T=(V,E)$, a vertex subset $S \subseteq V(T)$, weights $w: V(T) \rightarrow \mathbb{N}_{\ge 0}$}
    \Ensure{A subset of vertices $X$}
    \ForAll{$Q \in \binom{S}{\le30}$}
    \State $(I,O) = (Q, S \setminus Q)$ \Comment{\textcolor{blue}{Subset of $S$ (I)nside and (O)utside $F_Q$}}
    \If{$\exists a,b,c \in O$ such that $abc$ form a $\triangle$ in $T$}
      \State $F_Q=S$\Comment{\textcolor{blue}{A trivial solution since $Q$ cannot be extended}} 
      \State \textbf{continue}
    \EndIf
    \State $F_{Q} := I$
    \State $T_I := T - I$
    \While{$\exists a,b \in O, c \in V(T) \setminus (I \cup O)$ such that $abc$ form a $\triangle$ in $T_I$}
      \State $F_Q = F_Q \cup \{c\}$
      \State $T_I = T_I - \{c\}$
    \EndWhile    
    \State $E' := \bigcup\limits_{\substack{a \in O\\abc \textnormal{ form a } \triangle \textnormal{ in }T_I}} (b,c)$
    \State $F_Q = F_Q \cup \mathsf{FindVertexCover}((V(T),E'),w)$
    \EndFor
    \algstore{myalg_2}
\end{algorithmic}
\end{algorithm}  
\begin{algorithm}[H]                     
\begin{algorithmic} [1]                   
\algrestore{myalg_2}
    \If{$s \le 30$}
      \State Return a lightest set among $\{F_Q\}_{Q \subseteq \binom{S}{\le30}}$ acc. to $w$, say $X$
    \EndIf
    \State Let $L$ be a set of $\frac{s}{6}$ lightest vertices in $S$ acc. to $w$
    \State $F_{0} := L \cup \textsf{FindSFVS}\left(T - L, S \setminus L, \textsf{update3}(w,L)\right)$
    \ForAll{$i \in [28]$}
      \State Pick $x_i$ uniformly at random from $S$
      \If{$\min\{d^+_{T[S]}(x_i), d^-_{T[S]}(x_i)\} < \frac{s}{18}+\frac1{4}-\frac12$}
        \State $F_i:= S$
        \State \textbf{continue}
      \EndIf
      \State $C_i = \{\}$
      \State $w_i = w$
\While{there is a $\triangle$ in $T - C_i$ contains $x_i$}\Comment{\textcolor{blue}{Eliminating $\triangle$'s that $x_i$ is part of}}
        \State Let $C'$ be a $\triangle$ in $T - C_i$ containing $x_i$ 
        \State Let $v$ be a lightest vertex in $C' \setminus \{x_i\}$ acc. to $w_i$
        \State $C_i = C_i \cup \{v\}$
        \State $w_i = \textsf{update4}(w_i, C' \setminus \{x_i\})$  
      \EndWhile
      \State Let $R_i^1 = R_{T - C_i}(x_i)$ and $R_i^2 = \overline{R}_{T - C_i}(x_i)$ \label{lst:line:blah2}
      \State $F_i$ := $C_i \cup$ \textsf{FindSFVS}($T[R_i^1], S \cap R_i^1, w_i$) $\cup$ \textsf{FindSFVS}($T[R_i^2], S \cap R_i^2, w_i$)
    \EndFor   
    \State Return a lightest set among $\{F_Q\}_{Q \in \binom{S}{\le30}}, F_0,F_1, \dots, F_{28}$ acc. to $w$, say $X$
  \end{algorithmic}
\end{algorithm}

\end{document}